\documentclass{amsart}

\usepackage{amssymb,latexsym,amsfonts,amsmath}
\usepackage{stmaryrd}
\usepackage{xcolor,colortbl}
\usepackage{multirow}
\usepackage{graphicx}
\usepackage{epstopdf}
\usepackage{etex}
\usepackage{epsfig}
\usepackage{psfrag}
\usepackage{subfigure}
\usepackage{url}
\usepackage{comment}
\usepackage{graphics}
\usepackage{graphicx, array, blindtext} 
\usepackage{pifont}
\newcommand{\cmark}{\ding{52}}%
\newcommand{\xmark}{\ding{55}}%
\usepackage{mathtools}

\usepackage{lipsum}
\usepackage[ruled,linesnumbered]{algorithm2e} 

\SetAlFnt{\small}
\SetAlCapFnt{\small}
\SetAlCapNameFnt{\small}
\SetAlCapHSkip{0pt}
\IncMargin{-\parindent}

\newtheorem{theorem}{Theorem}[section]

\newtheorem{definition}[theorem]{Definition}
\newtheorem{remark}[theorem]{Remark}
\newcommand{\Ze}{{\mathbb Z}}
\newcommand{\R}{{\mathbb{R}}}
\newcommand{\N}{{\mathbb{N}}}
\newcommand{\segcc}[1]{\ensuremath{{\left\llbracket#1\right\rrbracket}}}
\newcommand{\viz}{{\it viz.,}\xspace}
\newcommand{\ie}{{\it i.e.,}\xspace}

\usepackage{upgreek}
\usepackage{dsfont}

\begin{document}

\begin{abstract}
The paper addresses the issue of reliability of complex embedded control systems in the safety-critical environment. In this paper, we propose a novel approach to design controller that (i) guarantees the safety of nonlinear physical systems, (ii) enables safe system restart during runtime, and (iii) allows the use of complex, unverified controllers ({\it e.g.}, neural networks) that drive the physical systems towards complex specifications. We use abstraction-based controller synthesis approach to design formally verified controller that provides application and system-level fault tolerance along with safety guarantee. Moreover, our approach is implementable using commercial-off-the-shelf (COTS) processing unit. To demonstrate the efficacy of our solution and to verify the safety of the system under various types of faults injected in applications and in underlying real-time operating system (RTOS), we implemented the proposed controller for the inverted pendulum and three degree-of-freedom (3-DOF) helicopter. 
\end{abstract}

\title[Software Fault Tolerance for Cyber-Physical Systems]{Software Fault Tolerance for Cyber-Physical Systems via Full System Restart}

\author[P. Jagtap]{Pushpak Jagtap$^1$} 
\author[F. Abdi]{Fardin Abdi$^2$} 
\author[M. Rungger]{Matthias Rungger$^3$}
\author[M. Zamani]{Majid Zamani$^1$} 
\author[M. Caccamo]{Marco Caccamo$^4$}

\address{$^1$Department of Electrical and Computer Engineering, Technical University of Munich, Germany.}
\email{\{pushpak.jagtap,zamani\}@tum.de}
\address{$^2$Uber, USA.}
\address{$^3$Corporate Research Center, ABB, Switzerland.}
\address{$^4$ Department of Mechanical Engineering, Technical University of Munich, Germany.}
\maketitle

\section{Introduction}
With the increased use of embedded systems in various safety-critical environments, these systems are expected to provide high-performance with reliability. Delivering high-performance drives the need for more complex systems\footnote{Such systems usually use the operating system~(OS) primitives to perform I/O or to set up concurrent threads~\cite{lee2008cyber}, utilize vendor developed drivers and use open source libraries~\cite{sulaman2014development}.}. As a consequence, it increases the possibilities for errors and makes formal verification more challenging. 

In the safety-critical environment, the use of complex embedded control systems where one or more control applications run on top of the real-time operating system (RTOS) and share the resources may lead to safety violations\footnote{In this paper, safety means not to violate the constraints of the physical components.} due to various software level faults. There are two main root causes of such faults. First, the control application may issue a set of unsafe commands due to the incorrect logic (bugs) or fail to generate any commands at all~(referred to as \emph{application-level faults}). Second, even with a bug-free control application, faults in underlying software layers such as the RTOS can disrupt the execution of the controller and jeopardize the safety~(usually referred to as \emph{system-level faults}). Ideally, all the components of these systems including the RTOS must be formally verified to ensure that they are fault-free. However, due to the high complexity, formal verification of the entire platform is very difficult. Therefore, designing architectures that enable the system designers to utilize components such as RTOS and vendor drivers without requiring to prove their correctness to guarantee safety is very important. 

In this work, we proposed a novel approach to design controller that provides safety guarantee on the physical component in presence of application-level and system-level faults. Moreover, the proposed solution provides \emph{fault-tolerance} and \emph{liveliness} guarantees only using one commercial-off-the-shelf (COTS) computing platform. The proposed approach uses full system restart to recover from such application and system-level faults. However, restarting in the safety-critical environment is very challenging. In our previous work \cite{7945009}, we provided a solution for designing controllers ensuring fault-tolerance and safety for linear physical systems. Designing such controllers become very challenging if the physical components exhibit nonlinear dynamics (which is the case in most of real applications). To address this, in this paper we provide a procedure for the synthesis of abstraction-based correct-by-construction controllers for nonlinear physical systems that enables the entire computing system to be safely restarted at runtime. This controller can keep the nonlinear control system inside a subset of safety region, only by updating the actuator input at least once after every system restart. In this paper, we refer to this controller as \emph{Base Controller}~(BC). 

Restarting a system is an effective approach for recovery from unknown faults at runtime with a very predictable outcome. As soon as a fault occurs that disrupts the execution of critical software components, a hardware watchdog timer~(WD) restarts the system. During a restart, a fresh image of all the software~(middleware, RTOS, and applications) is loaded from a read-only storage which recovers the system into an operational state. Prior to this work, restarting was proposed as a way to increase the availability of non-safety critical systems~\cite{candea2001recursive,Candea03crash-onlysoftware,candea2003jagr,candea2004microreboot,vaidyanathan2005comprehensive,garg1995analysis,huang1995software}. Alongside, partial restarting of safety-critical systems using extra hardware was investigated in~\cite{fardin2016reset,bak2009system}. To the best of our knowledge, this is the first work that proposes safe restarting of the entire system in a safety-critical environment containing nonlinear physical components.

Having only BC and the WD mechanism which enables restarting, allows the system to remain safe, tolerate faults and recover from them. However, it does not make any progress towards its mission goal. To address this issue, BC is complemented with a \emph{Mission Controller}~(MC) ({\it e.g.}, a neural network) and a \emph{Decision Module}~(DM). The MC is an unverified, high-performance, complex controller that drives the system towards the mission setpoints. It may contain unsafe logic or bugs that jeopardize safety. To maximize the progress towards the mission goals, in every control cycle, DM checks the MC command. If it satisfies the safety requirements, DM allows it to be sent to the actuators. Otherwise, BC command is applied to the system. By doing so, MC drives the system for as long as possible, and, whenever it is not possible, BC takes the control. The logical view of this design is depicted in Figure~\ref{fig:architecture}. 

In the proposed design, the only components that need to be verified for correct functionality are BC, DM and Flushing Task. Any fault in the system software~(System-Level or Application-Level) that results in a fail-silent failure~(also known as fail-stop) of these two components leads to WD triggering a system-wide restart and recovery. However, our design does not protect the system from faults that alter the logic of BC or DM at execution times. In summary, this design enables the system to provide formal safety guarantees by verifying only the correctness of BC, DM, and Flushing Task instead of entire MC, RTOS, and middleware.

The key contributions of our work are:
\begin{itemize}
	\item Construction of formally verified base controllers for safety-critical applications with nonlinear physical components which guarantee safe full system restart for application and system level fault tolerance. 
	\item Tolerating application-level faults as well as system-level faults using only one COTS processing unit.
	\item Empirical validation of both the practicality of our proposed design and the safety guarantees through fault-injection testing on a prototype controller for the nonlinear inverted pendulum system and a 3-DOF helicopter.
\end{itemize}


\begin{figure}[ht]
	\begin{center}
		\includegraphics[width=0.7\textwidth]{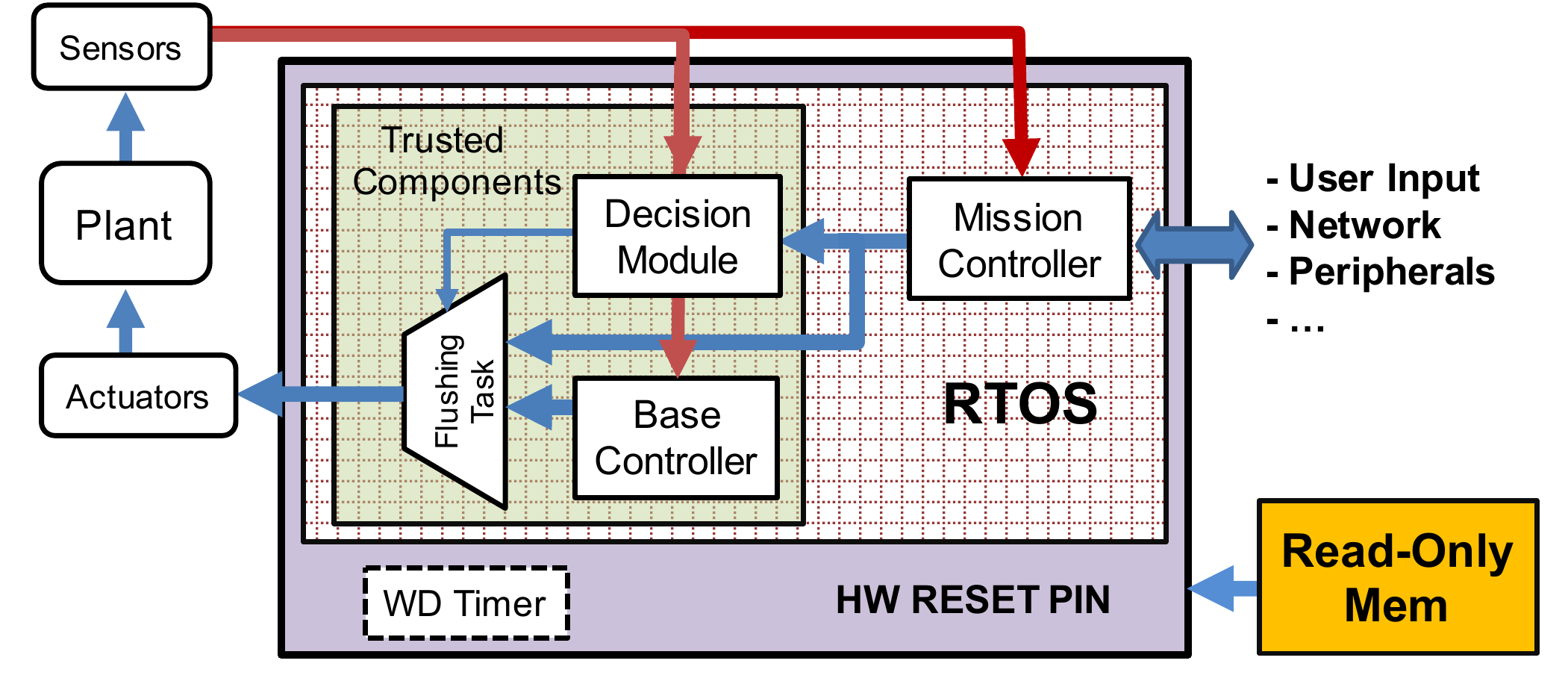}
		\caption{The logical view of the proposed design.}
		\label{fig:architecture}
	\end{center}
\end{figure}

\section{Related Work}
The concept of utilizing an unverified, complex controller along with a simple, verified safety controller for fault tolerance was initially proposed as Simplex Architecture in~\cite{sha1998dependable, Sha01usingsimplicity, sha1996evolving}. In earlier simplex designs, fault tolerance was achieved in one of two ways. In some of these designs such as~\cite{sha1998dependable,seto1999engineering,Sha01usingsimplicity,sha1996evolving,crenshaw2007simplex}, all three components (safety controller, complex controller, and decision unit) share the same computing hardware (processor) and software platform (OS, middleware). As a result, these designs only protect the safety against the faults in the application logic of the complex controller. And, there is no guarantee of the correct behavior in the presence of system-level faults. Our proposed design protects the system from both application-level and system-level faults.

Some Simplex-based designs such as System-Level Simplex~\cite{bak2009system}, Secure System Simplex Architecture (S3A)\cite{mohan2013s3a}, and other variants \cite{vivekanandan2016simplex} run the safety controller and the decision logic on an isolated, dedicated hardware unit. By doing so, the trusted components are protected from the faults in the complex subsystem. However, exercising System-Level Simplex design on most COTS multicore platforms/SoC~(system on chip) is challenging. The majority of commercial multicore platforms is not designed to achieve strong inter-core fault isolation due to the high-degree of hardware resource sharing. For instance, a fault occurring in a core with the highest privilege level may compromise power and clock configurations of the entire platform. To achieve full isolation and independence, one has to utilize two separate boards/systems. In contrast, the approach proposed in this paper needs only one processor and tolerates system-level faults. 

Note that, the control domain of the proposed BC depends on the system dynamics and the restart time of the platform. Increased restart time, shrinks the domain of the BC. For a given system, it may be empty; meaning that the dynamics of the system does not allow a controller with such properties to exist. System-Level Simplex does not have this limitation because it uses a dedicated hardware that is not impacted by faults~(or restarts) in the complex controller unit. The proposed approach is especially suitable for the Internet of Things~(IoT) applications, requiring increased robustness at low cost and without using extra hardware as System-Level Simplex requires.

The notion of restarting as a means of recovery from faults and improving system availability is previously proposed in the literature. These approaches are generally divided into two categories, \viz \textit{i}) \emph{revival}, reactively restart a failed component and \textit{ii}) \emph{rejuvenation}, prophylactically restart functioning components to prevent state degradation~\cite{candea2004improving}. Our approach, as described in this paper, fits in the former category. However, with slight modification, our design can incorporate periodic self-triggered restarts to prevent future unscheduled unavailable times. In the second form, this work can also be categorized in the latter category. 

Most of the previous works on restarting are proposed for traditional \textit{non} safety-critical computing systems such as servers and switches. Authors in \cite{candea2001recursive} introduce recursively restartable systems as a design paradigm for highly available systems and use a combination of revival and rejuvenation techniques. Earlier literature \cite{Candea03crash-onlysoftware,candea2003jagr,candea2004microreboot} illustrates the concept of microreboot which consists of having fine-grain rebootable components and trying to restart them from the smallest component to the biggest one in the presence of faults. The works in~\cite{vaidyanathan2005comprehensive,garg1995analysis,huang1995software} focus on failure and fault modeling and try to find an optimal rejuvenation strategy for various systems. In this context, our previous work in Reset-Based Recovery~\cite{fardin2016reset} was an attempt to utilize restarting as a recovery method for computing systems in safety-critical environments. In \cite{fardin2016reset}, we used System-Level Simplex architecture and proposed to restart only the complex subsystem upon the occurrence of faults. This is feasible because the safety subsystem runs on a dedicated hardware unit and is not impacted by the restarts in the complex subsystem. The approach of the current paper is significantly different and uses only one hardware unit.

\section{Control System Description}
\label{sec:safety}
\subsection{Notations}
The symbols $ \N $, $ \N_0 $, $\Ze$, $ \R$, $\R^+,$ and $\R_0^+ $ denote the set of natural, nonnegative integer, integer, real, positive, and nonnegative real numbers, respectively. We use $ \R^{n\times m} $ to denote a vector space of real matrices with $ n $ rows and $ m $ columns. The identity matrix in $\R^{n\times n}$ is denoted by $I_n$ and zero matrix in $R^{n\times m}$ is denoted by $0_{n\times m}$. For $a,b\in(\R\cup\{-\infty, \infty\})^n$, $a\leq b$ component-wise, the closed hyper-interval is denoted by $\segcc{a,b}:=\R^n\cap ([a_1,b_1]\times [a_2,b_2]\times\ldots\times [a_n,b_n])$. We identify the relation $R\subseteq A\times B$ with the map $R:A\rightarrow2^B$ defined by $b\in R(a)$ iff $(a,b)\in R$. Given a relation $R\subseteq A\times B$. $R^{-1}=\{(b,a)\in B\times A\mid (a,b)\in R\}$. $Q\circ R$ denotes the composition of maps $Q$ and $R$, $Q\circ R(x)=Q(R(x))$. The map $R$ is said to be strict when $R(a)\ne\emptyset$ for every $a\in A$. 
\subsection{Nonlinear Control Systems}
\begin{definition}[Nonlinear control systems]
	A nonlinear control system is a tuple $\Sigma=(\mathbb{R}^n,\mathsf{U},\mathcal{U},f)$, where $\mathbb{R}^n$ is the state space; $\mathsf{U}\subseteq\mathbb{R}^p$ is a bounded input set; $\mathcal{U}$ is a subset of the set of all functions of time from $\mathbb{R}^+_0$ to $\mathsf{U}$; and $f$ is a locally Lipschitz continuous map from $\mathbb{R}^n\times\mathsf{U}$ to $\mathbb{R}^n$.
\end{definition}
The trajectory $\xi$ is said to be a solution of $\Sigma$ if there exists $\upsilon\in\mathcal{U}$ satisfying:
\begin{align}
\dot{\xi}(t)=f(\xi(t),\upsilon(t)),
\label{sys}
\end{align}
for any $t\in\mathbb{R}^+_0$. We emphasize that the locally Lipschitz continuity assumption on $f$ ensures existence and uniqueness of solution $\xi$ \cite{sontag2013mathematical}. We use notation $\xi_{x,\upsilon}(t)$ to denote the value of solution at time $t$ under the input signal $\upsilon$ and starting from initial state $x=\xi_{x,\upsilon}(0)$. 
%
\subsection{Formulating Safety}
\label{sec:safetydefinition}
In physical systems, maintaining all states and control inputs within safe limits is very important in order to avoid damages to the system itself or the environment around it. 
In this paper, we define safety region $\mathcal{S}$ as a subset of the state space. For example, one can define it as: 
\begin{itemize}
	\item polytope $\mathcal{S}= \{x\in\R^n\mid H_x \cdot x \leq h_x \}$ parameterized by $H_x\in\R^{q\times n}$, $h_x\in\R^q$, or
	\item ellipsoid $\mathcal{S}= \{x\in\R^n\mid \|L(x-y)\|_2\leq1 \}$ parameterized by $L\in\R^{n\times n}$ and $y\in\R^n$.
\end{itemize}
In a similar way, the bounds on operational ranges of control inputs can be expressed as:
\begin{itemize}
	\item polytope $\mathcal{S}_u = \{u \in\mathsf{U}\mid H_u \cdot u \leq h_u\}$ parameterized by $H_u\in\R^{\bar q\times p}$ and $h_u\in\R^{\bar q}$, or
	\item ellipsoid $\mathcal{S}_u= \{u\in\mathsf{U}\mid \|L_u(u-\overline u)\|_2\leq1 \}$ parameterized by $L_u\in\R^{p\times p}$ and $\overline u\in\mathsf{U}$.
\end{itemize}
The nonlinear control system $\Sigma$ is said to be safe if the states of the system remain inside $\mathcal{S}$ using only the control commands in $\mathcal{S}_u$.

\subsection{Reachable Set}
Consider a nonlinear control system as in (\ref{sys}) and a set $X_0\subset R^n$. The reachable set of states that can be reached starting from set $X_0$ under input signal $\upsilon$ at time $\uptau$ is given by $\mathbf{Reach}_{\uptau}(X_0,\upsilon):=\bigcup_{x\in X_0}\xi_{x,\upsilon}(\uptau)$. We use notation $\mathbf{Reach}_{[0,\uptau]}(X_0,\upsilon)$ to denote the reachable set that can be reached starting from $X_0$ under input signal $\upsilon$ up to time $\uptau$ and can be defined as   $\mathbf{Reach}_{[0,\uptau]}(X_0,\upsilon):=\bigcup_{t\in [0,\uptau]}\mathbf{Reach}_{t}(X_0,\upsilon)$. We use the notation $\overline{\mathbf{Reach}}_{\uptau}(X_0,\upsilon)$ to denote an over-approximation of the set $\mathbf{Reach}_\uptau(X_0,\upsilon)$.

\section{Design Approach}
\label{sec:methodology}


As depicted in Figure~\ref{fig:architecture}, the proposed design consists of three main components; Base Controller~(BC), Mission Controller~(MC) and Decision Module~(DM). 

The BC is a verified, reliable controller that is only concerned with safety. 
It does not make progress towards the mission set points of the system~(\ie it does not provide \emph{liveness}). The MC, on the other hand, is the main controller which is concerned with the mission-critical requirements. This controller may have complex logic, can be changed and upgraded while the system is running and may even contain unsafe logic and bugs. 
As an example, MC may be a neural network resulted from machine learning techniques.

All the components of the system run on top of the RTOS. The length of one control cycle of the system is $\uptau_c$. The $k$th control cycle refers to the period $[(k-1)\uptau_c, k\uptau_c]$, where $k \in \mathbb{N}$. The cycles count and the time origin are restarted after every system restart. Therefore, $k = 1$ always refers to the first cycle after the latest system restart. Furthermore, we assume that the length of the restart time\footnote{It includes the time for reloading the bootloader, OS, and the applications from the read-only storage, initializing the necessary sensors and peripheral, booting the OS and executing the control applications.}, $\ie$ $\uptau_r$, of the system is an integer multiple\footnote{ Restart time can be rounded up to match the closest $k\uptau_c$.} of $\uptau_c$ (\ie $\uptau_r = m \uptau_c$, where $m \in \mathbb{N}$). While the system is running, sensor values are sampled at $t=k \uptau_c - \epsilon$ where $\epsilon \ll \uptau_c$ and actuator inputs are updated at $t=k\uptau_c$.

In every control cycle, after MC runs and generates its output $u_{mc}$, DM evaluates the safety requirements under $u_{mc}$ and decides whether $u_{mc}$ can be applied to the actuators. Then, DM writes its output, along with the corresponding MC command and a timestamp~(cycle number) to a fixed memory address.

At the end of the control cycle, at time $k\uptau_c - \epsilon$ after sensors are sampled, BC runs and generates $u_{bc}$. Then a flushing task retrieves $u_{mc}$, $u_{bc}$, the decision of DM and the corresponding timestamp from the memory. If the timestamp matches with the current cycle number, $k$, it updates the actuator commands with $u_{mc}$ or $u_{bc}$ based on the decision of DM and resets watchdog timer (WD). Non-matching timestamps indicate that one or both of the DM and BC tasks did not execute or missed their deadlines. In such cases, the flushing task does not update the WD. Consequently, WD expires at $t = k\uptau_c$ and triggers a restart. Note that as a result of this mechanism, restarts are only triggered at times $t = k\uptau_c$ and do not occur in between control cycles. The steps are illustrated in Figure~\ref{fig:cycles}.



\begin{figure}[ht]
	\begin{center}
		\includegraphics[width=0.7\textwidth]{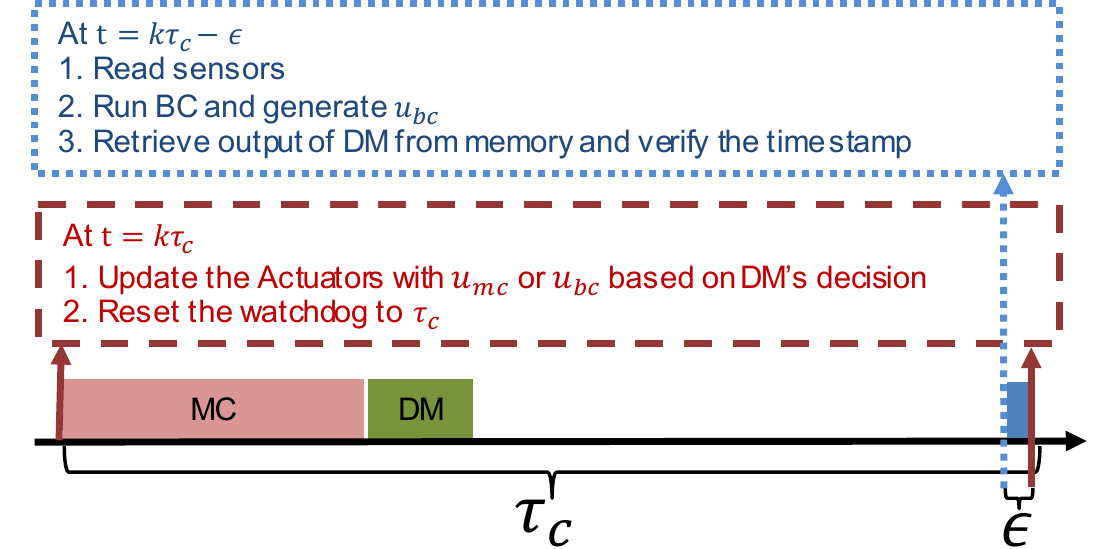}
		\caption{Sequence of events within one control cycle.}
		\label{fig:cycles}
	\end{center}
\end{figure}

In the rest of this section, we discuss the assumptions and the fault model of the system. Then, we introduce the properties of the BC and how it is able to safely tolerate the restarts. Finally, we discuss the safe switching logic of the DM.


\subsection{Assumptions and Fault Model}
\label{sec:assumptions}
In this section, we clarify several assumptions we make about faults and components of the system.

\begin{itemize}
	\setlength\itemsep{0em}
	\item 
	In this work, we are not concerned about hardware faults and we assume that hardware is reliable.
	\item BC, DM, and flushing task are independently verified and fault-free. They might, however, fail silently~(no output is generated) due to the faults in the previously dependent software layers or other applications.
	\item System-level and application-level faults may cause BC, DM, and flushing task to fail silently but may not change their logic or alter their output. 
	\item Once a command is sent to an actuator input; the actuator holds that value until the control system sends a new actuation command. Therefore, during a system-level restart, the actuators operate with the last command that was sent before the restart occurred\footnote{Commercial chips such as~\cite{nxpPCA9685} are available that provide programmable PWM controller. Using these intermediate chips one can prevent the invalid signals that may appear on the general-purpose input/output (GPIO) port of the board during a restart, from changing the actuation command.}.
	
	\item We assume that the system-level faults do not happen within the first $\uptau_r$ seconds after the boot is complete so that the BC has the chance to execute correctly at least once. In other words, this assumption implies that the system is not completely dysfunctional. In Section~\ref{sec:properties}, we demonstrate the necessity of this assumption.
	
	
\end{itemize}


\subsection{Properties of the Base Controller}
\label{sec:properties}

In this section, we provide the properties required for the BC as follow: \\
There exists a subset $\mathcal{I}$ of the state space, such that  for all $x \in \mathcal{I}$ at time $t_0\in\R^+_0$, there exists a control
command $u_{bc}\in \mathcal{S}_u$, such that:
\begin{itemize}
	\item[(i)] $\xi_{x,u_{bc}}(t_0 + \uptau_c ) \in \mathcal{I}$,
	\item[(ii)] $\xi_{x,u_{bc}}(t_0 + \uptau_c + \uptau_r) \in \mathcal{I}$, and 
	\item[(iii)] $\xi_{x,u_{bc}}(t) \in \mathcal{S}$ for $t\in [t_0,t_0+\uptau_c + \uptau_r]$. 
\end{itemize}
Note that, in the rest of the paper we assumed that the actuators hold the control input constant within the period of $[t_0, t_0+\uptau_c+\uptau_r]$.

Intuitively, above properties imply that if the current state of the system is inside $\mathcal{I}$, BC is able to generate a control command that keeps the physical system safe. For the intuition, consider $t_0 = k\uptau_c$. Property (i) implies that one control cycle after $u_{bc}$ is applied to the actuators, at the end of ($k+1$)th cycle, state is inside $\mathcal{I}$. Therefore, if the system is still running and no faults have occurred, BC is able to find another safe command at $t = (k+1)\uptau_c$. If a fault had occurred within the ($k+1$)th cycle, a restart will be triggered at the end of the cycle and BC will not be available to update the actuator input. Property (ii) implies that in such a case, the system will be in $\mathcal{I}$, after the restart completes. This guarantees that the system can be kept safe after the restart completes. Finally, property (iii) ensures that the system remains inside the safety region during ($k+1$)th cycle and a possible consequent restart.

A BC with the above properties, without any other components, can keep the system safe, only if it updates the actuator commands at least once after every restart $\uptau_r$. Therefore, it is necessary for the system to not have any system-level faults within the first $\uptau_r$ seconds after the restart.

\subsection{Switching Logic of DM}
\label{sec:runtime}

A system with only BC remains safe and tolerates restarts but it does not make any progress towards the mission goal. In order to maximize the progress towards the mission goal, it is desirable to use the MC command in every cycle whenever it is possible.




In every cycle $k$, DM runs and evaluates the following conditions. If those  conditions hold, $u_{mc}$ is safe to be applied to the actuator inputs at the end of the cycle~(\ie at time $t = k\uptau_c$). Otherwise, DM chooses $u_{bc}$. Following conditions guarantee that the system remains safe and recoverable under $u_{mc}$ whether it restarts or not. 

\begin{itemize}
	\item[(i)]$\mathbf{Reach}_{\uptau_c}(\bar{x}[k], u_{mc}) \subseteq \mathcal{I}$
	\item[(ii)]$\mathbf{Reach}_{\uptau_r+\uptau_c}(\bar{x}[k], u_{mc}) \subseteq \mathcal{I}$
	\item[(iii)]$\mathbf{Reach}_{[0, \uptau_r+\uptau_c]}(\bar{x}[k], u_{mc}) \subseteq \mathcal{S}$
\end{itemize}

Here, $\uptau_r$ and $\uptau_c$ are the length of the restart time and of the control cycle of the platform. Notation $\bar{x}[k]$ denotes the state of the system when the actuator command is going to be applied to the system~(\ie the end of the cycle at time $t=k\uptau_c$).

From properties of the BC, it is known that if the state is inside $\mathcal{I}$, BC can find a control command that keeps the system in safe and restartable region. Condition~(i) ensures that one control cycle after $u_{mc}$ is applied to the system the state will be inside $\mathcal{I}$. If no faults occur within the control cycle, BC is guaranteed to be able to find a safe control for the system. However, if a fault occurs within the cycle, WD triggers a system restart at the end of the cycle. Condition~(ii) ensures that state will be inside $\mathcal{I}$ when the restart completes~(\ie at $\uptau_c+\uptau_r$). Furthermore, condition~(iii) guarantees that during the control cycle and restart time (if it happens) state remains inside the safety region.

Note that, in the real implementation, calculating reachable set and therefore evaluating these conditions requires time and does not happen instantaneously. Therefore, assuming $k$ is the current cycle, above conditions have to be assessed before $t=k\uptau_c$. At this time, however, $x[k]=x(k\uptau_c)$~(state of the system when the actuator command is going to be updated) is not available yet. To address this issue, above conditions use $\bar{x}[k]$ which is the over-approximated prediction of $x[k]$ based on $x[k-1]$~(sampled sensor values in the previous cycle). Prediction $\bar{x}[k]$ can be computed in the following way:
\[
\bar{x}[k] = \overline{\mathbf{Reach}}_{\uptau_c}(x[k-1], u_{k-1}),
\]
where $x[k-1]$ is the sampled state at the previous cycle~(state of the system at the beginning of the current control cycle). Input $u_{k-1}$ is the control command sent to the actuators in the previous cycle. Since, in the first control cycle after a restart, $u_{\text{k-1}}$ is not available, the DM always chooses the BC in the first cycle. To compute an over-approximation of reachable set for nonlinear control systems there are various approaches available in literature for example see \cite[Section VIII.c]{ReissigWeberRungger17}, \cite{6632887}, and \cite{asarin2003reachability}.


\section{Base Controller Design}
\label{BC}
In this section, we provide a systematic approach to design base controllers ensuring properties mentioned in Subsection \ref{sec:properties}. To design BC, we use symbolic controller synthesis approach which uses the discrete abstractions of nonlinear physical systems \cite{tabuada2009verification}. The advantage of using this approach is that it provides formally verified controllers for high-level specifications (usually expressed as linear temporal logic (LTL) formulae \cite{baier2008principles}). One can readily see that the properties given in Subsection \ref{sec:properties} are equivalent to invariance specification.
\subsection{Transition Systems and Equivalence Relation}
We recall the notion of \textit{transition system} introduced in \cite{tabuada2009verification} which will later be used as unified framework to represent nonlinear control systems and corresponding discrete abstractions.
\begin{definition}[Transition system]
	A transition system is a tuple $S=(X,X_0,U,\longrightarrow)$ where $X$ is a set of states, $X_0\subseteq X$ is a set of initial states, $U$ is a set of inputs, $\longrightarrow\subseteq X\times U \times X$ is a transition relation.
\end{definition}
We denote by $x\overset{u}{\longrightarrow} x'$ an alternative representation for transition $(x,u,x')\in\longrightarrow$, where state $x'$ is called a $u$-successor (or simply successor) of state x , for some input $u\in U$. We denote by $Post_u(x)$ the set of all $u$-successors of state $x$, and by $U(x)$ the set of all admissible inputs $u\in U$ for which $Post_u(x)$ is non-empty. Now, we provide the notion of feedback refinement relation between two transition systems, introduced in \cite{ReissigWeberRungger17}, which is later used to construct discrete abstractions and base controllers for nonlinear control systems $\Sigma$. 
\begin{definition}[Feedback refinement relation]
	Consider two transition systems  $S_1=(X_1,X_{10}$, $U_1,\underset{1}{\longrightarrow})$ and $S_2=(X_2,X_{20},U_2,\underset{2}{\longrightarrow})$ with $U_2 \subseteq U_1$. A strict relation $Q \subseteq X_1\times X_2$ is a feedback refinement relation from $S_1$ to $S_2$ if following conditions hold for every pair $(x_1,x_2) \in Q$:
	\begin{itemize}
		\item[(i)] $U_2(x_2) \subseteq U_1(x_1)$,
		\item[(ii)] $u\in U_2(x_2) \Rightarrow Q(Post_u(x_1))\subseteq Post_u(x_2)$,
	\end{itemize}
	and the feedback refinement relation from $S_1$ to $S_2$ is denoted by $S_1\preceq_Q S_2$.
\end{definition}
Intuitively, the above relation says that all admissible inputs of $S_2$ can be used in transition system $S_1$ such that all transitions in $S_1$ are associated with corresponding transitions in $S_2$. As a result, one can easily refine controller synthesized for $S_2$ using feedback refinement relation $Q$ to make it compatible for $S_1$. Further details about feedback refinement relation and its role in the controller synthesis can be found in \cite{ReissigWeberRungger17}.  
\subsection{Sampled-Data Control System as a Transition System} 
As we discussed in the previous sections, the sampling time can take any value in $h=\{\uptau_c,\uptau_r+\uptau_c\}$ depending on the occurrence of fault. We assume that the value of control input is held for the respective sampling period. The transition system associated with the nonlinear control system $\Sigma$ with such a sampling behavior can be given by the tuple
\begin{align}
S_h(\Sigma)=(X_h,X_{h 0},U_h,\underset{h}{\longrightarrow}),
\end{align} 
where 
\begin{itemize}
	\item $X_h=\R^n$, $X_{h 0}=\R^n$, $U_h=\mathsf{U}$, and 
	\item $x\overset{u}{\underset{h}{\longrightarrow}} x'$ is a transition if and only if there exists $x'= \xi_{x,u}(\uptau_c)$ or $x'= \xi_{x,u}(\uptau_r+\uptau_c)$, where $u\in U_h$.
\end{itemize}
Note that we abuse notation above by identifying $u$ with the constant input curve with domain $[0,\uptau_c]$ or $[0,\uptau_r+\uptau_c]$ and value $u$.

For the transition system $S_h(\Sigma)$, the finite or infinite run generated from initial state $x_0\in X_{h 0}$ is given by $x_0\overset{u_0}{\underset{h}{\longrightarrow}}x_1\overset{u_1}{\underset{h}{\longrightarrow}}x_2\overset{u_2}{\underset{h}{\longrightarrow}}\ldots$ such that $x_i \overset{u_i}{\underset{h}{\longrightarrow}} x_{i+1}'$, for $i\in \N_0$.

By considering properties of BC mentioned in Subsection \ref{sec:properties}, one can view it as a safety controller synthesis problem for $S_h(\Sigma)$. 
\begin{definition}[Safety controller]\label{def:safe_controller} Consider a safe set $\mathcal{S}\subseteq \R^n$ as given in Subsection \ref{sec:safetydefinition}, a \textit{safety controller} for $S_h(\Sigma)$ is given by a map $C_{h}:X_h\rightarrow 2^{U_h}$ such that:
	\begin{itemize}
		\item[(i)] for all $x\in X_h$,  $C_{h}(x)\subseteq U_h(x)$,
		\item[(ii)] its domain $dom(C_{h})=\{x \in X_h\mid C_{h}\ne \emptyset\}\subseteq\mathcal{S}$,
		\item[(iii)] for all $x\in dom(C_{h})$ and $u\in C_{h}(x)$, $Post_u(x)\subseteq dom(C_{h})$.
	\end{itemize}
\end{definition}
Essentially, a safety controller generates infinite runs $x_0\overset{u_0}{\underset{h}{\longrightarrow}}x_1\overset{u_1}{\underset{h}{\longrightarrow}}x_2\overset{u_2}{\underset{h}{\longrightarrow}}\ldots$  such that $x_i\in\mathcal{S}$, for all $i\in \N_0$. At the end of this section, we provide a systematic way to compute such controller for linear control systems. However, finding such a control strategy for complex nonlinear control systems is quiet difficult. This motivates the use of abstraction-based synthesis methods described below.\\
\subsection{Discrete Abstraction}\label{Disc_abs}
To design controllers for the concrete system $S_h(\Sigma)$ from its abstraction, the system and its abstraction must satisfy formal behavioural inclusions in terms of feedback refinement relations. Consider sampling times $\uptau_c,\uptau_r+\uptau_c\in\R^+$ and quantization parameter $\eta\in(\R^+)^n$. The \textit{discrete abstraction} of $S_h(\Sigma)$ is given by the tuple 
\begin{align} 
S_q(\Sigma)=(X_q,X_{q 0},U_q,\underset{q}{\longrightarrow}),
\end{align}
where
\begin{itemize}
	\item $X_q$ is a cover of $X_h$ and elements of the cover $X_q$ are nonempty, closed hyper-intervals referred to as cells. For computation of the abstraction, we consider subset $\overline{X}_q  \subseteq X_q$ of congruent hyper-rectangles aligned on a uniform grid parameterized with quantization parameter $\eta\in (\R^+)^n$ and given by
	$\eta\mathbb{Z}^n=\{c\in\R^n\mid\exists_{k\in\mathbb{Z}^n}\forall_{i\in\{1,2,\ldots,n\}}c_i=k_i\eta_i\}$, \ie $x_q\in\overline{X}_q$ implies that there exists $c\in\eta\mathbb{Z}^n$ with $x_q=c+\segcc{\frac{\eta}{2},\frac{\eta}{2}}$. The remaining cells $X_q\setminus\overline{X}_q$ are  considered as "overflow" symbols, see \cite[Sec III.A]{5770194}
	\item $X_{q0}\subseteq X_q$, 
	\item $U_q$ is a finite subset of $U_h$, 
	\item for $x_q\in\overline{X}_q$ and $u\in U_q$, define $A:=\{x_q'\in X_q \mid (x_q'\cap\overline{\mathbf{Reach}}_{\uptau_c}(x_q,u_q))\cup(x_q'\cap\overline{\mathbf{Reach}}_{\uptau_r+\uptau_c}(x_q,u_q))\ne\emptyset\}$. If $A\subseteq \overline X_q$, then $Post_u(x_q)=A$, and otherwise $Post_u(x_q)=\emptyset$. Moreover, $Post_u(x_q)=\emptyset$ for all $x_q\in X_q\setminus\overline{X}_q$.
\end{itemize}
For the exact procedure to compute such discrete abstraction, we refer interested readers to \cite{rungger2016scots}. 
\begin{theorem}
	If $S_q(\Sigma)$ is a discrete abstraction of $S_h(\Sigma)$ with sampling times $\uptau_c, \uptau_r+\uptau_c\in\R^+$, and quantization  parameter $\eta\in(\R^+)^n$, then $S_h(\Sigma)\preceq_Q S_q(\Sigma)$.
\end{theorem}
\begin{proof}
	The proof is similar to the proof of \cite[Theorem VIII.4]{ReissigWeberRungger17}.
\end{proof}
The abstract safe set $\hat{\mathcal{S}}$ for $S_q(\Sigma)$ is given by $\hat{\mathcal{S}}:=\{x_q\in \overline{X}_q\mid Q^{-1}(x_q)\subseteq\mathcal{S}\}$.
\subsection{Controller Synthesis and Refinement}
In this section, we consider the problem of synthesis of safety controller $C_h$ for $S_h(\Sigma)$ and safe set $\mathcal{S}$. Because of the feedback refinement relation, we can solve safety controller synthesis problem for the discrete abstraction $S_q(\Sigma)$ and abstract safe set $\hat{\mathcal{S}}$. Let $C_q:X_q\rightarrow U_q$ be the maximal safety controller satisfying conditions in Definition \ref{def:safe_controller} for $S_q(\Sigma)$ and safe set $\hat{\mathcal{S}}$. Since $S_q(\Sigma)$ has finite states and inputs, we can use standard maximal fixed-point computation algorithm \cite{tabuada2009verification} for the computation of $C_q$. One can easily refine this controller for $S_h(\Sigma)$ and safe set $\mathcal{S}$ using the following theorem:
\begin{theorem}
	If $S_h(\Sigma)\preceq S_q(\Sigma)$ and $C_q$ is the safety controller for $S_q(\Sigma)$ and $\hat{\mathcal{S}}$, then the refined controller $C_h:=C_q\circ Q$ solves the safety problem for $S_h(\Sigma)$ and $\mathcal{S}$.
\end{theorem} 
\begin{proof}
	The proof is similar to the proof of \cite[Theorem VI.3]{ReissigWeberRungger17}.
\end{proof}
Intuitively, the refined controller $C_h$ for $S_h$ can naturally be obtained from the abstract controller $C_q$ by using the feedback refinement relation $Q$ as a \textit{quantizer} to map $x_h$ to $x_q\in Q(x_h)$.   
\begin{remark}
	The obtained controller $C_h$ solves the safety problem for the sampled system, i.e., the obtained base controller satisfies the first two properties mentioned in Subsection \ref{sec:properties} with invariant set $\mathcal{I}=dom(C_h)$. However, one can ensure safety guarantee of inter-sampling trajectory (\ie third property in Subsection \ref{sec:properties}) by shrinking the safe set by a magnitude computed using the global Lipschitz continuity property of map $f$. 
\end{remark}
Despite the applicability of the proposed approach for complex and nonlinear control systems, it suffers from the curse of dimensionality, i.e., the computational complexity increases exponentially with state-space dimensions of concrete systems. There are few results available to address this issue for some class of nonlinear control systems \cite{zamani2017towards,zamani2017compositional}. In next subsection, we provide an alternative approach to compute invariant set $\mathcal{I}$ and BC for linear control systems.  

\subsection{Base Controller for Linear Control Systems}
\label{BC_linear}
In this subsection, we provide an algorithm to compute $\mathcal{I}$ using discretized linear-control systems. The continuous linear control system can be converted to a discrete control system with the sampling time of $\uptau_c$ as: 
\begin{equation}
\dot{\xi}(t) = A\xi(t) + B\upsilon(t) \rightarrow x[k+1] = A_dx[k] + B_du[k],
\label{eq:dynamicdis}
\end{equation}
where $A_d = e^{A\uptau_c} = \sum_{k=0}^{\infty}  \frac{1}{k!}(A\uptau_c)^k \simeq \sum_{k=0}^{p}  \frac{1}{k!}(A\uptau_c)^k$, and $B_d = \left(\int_{0}^{\uptau_c}e^{At}dt \right) \cdot B$.


In this subsection, we show how to construct a BC with the properties: $\forall x[k] \in \mathcal{I}, \exists u_0$, where $u[p] = u_0, p\in\{k, k+1, ..., k+m\}$ such that 
(i) $x[k+1]\in \mathcal{I}$ and (ii) $ x[k+1+m] \in \mathcal{I}$, where $m = \uptau_r/\uptau_c$ and $m\in \mathbb{N}$. However, these properties does not guarantee that the inter-sample behaviour is in the safe region. To address this issue we need to readjust the safe region. For the systematic procedure to compute readjusted safe region $\mathcal{S}' \subseteq \mathcal{S}$, we refer interested reader to \cite[Section 5.1]{7945009}

\subsubsection{Finding the Invariant Subset $\mathcal{I}$}
\label{sec:iregion}


To compute the set $\mathcal{I}$, we closely follow the usual 
construction method based on backwards reachable sets to compute the
\textit{largest} invariant set for linear discrete-time systems~(see
e.g. in~\cite{blanchini2008set}).
We slightly modify this procedure and present it in Algorithm~\ref{alg:invariant} to compute the subset $\mathcal{I} \subseteq \mathcal{S}'$, such that for the discrete-time system in (\ref{eq:dynamicdis}), $\mathcal{I}$ satisfies the properties in Subsection~\ref{sec:properties}.




\begin{algorithm}[ht]
	\SetAlgoLined\DontPrintSemicolon
	\SetKwFunction{algo}{Switching Logic}
	
	\SetKwFunction{algo}{ComputeInvRegion}
	\SetKwProg{myalg}{}{}{}
	\myalg{\algo{$H^a_x$, $h^a_x$, $H_u$, $h_u$, $A_d$, $B_d$, $A^{(m+1)}_d$, $B^{(m+1)}_d$}}
	{	
		$I^{(0)}$ :=  Polytope($H^a_x \cdot x \leq  h^a_x$) and $p$ = 0\label{line:adjust} \\
		\While{$p < p_{max}$}{
			$[H'_x, h'_x]$ := PolytopeToMatrix($I^{(p)}$)\\
			pt := Polytope
			\[
			\left(
			\left[\begin{array}{c}
			H'_x\\    
			H'_x\\
			H_u
			\end{array}
			\right]
			\left[
			\begin{array}{cc}
			A_d^{(m+1)}& B_d^{(m+1)}\\    
			A_d& B_d\\
			0_{m \times n}&I_{m} \\
			\end{array}
			\right]
			\left[
			\begin{array}{c}
			x\\
			u\\
			\end{array}
			\right]   \leq 
			\left[
			\begin{array}{c}
			h'_x\\
			h_u\\
			\end{array}
			\right] 
			\right)
			\] \label{line:stateexpansion}\\
			$I^{(p+1)}$ := pt.projectOnStateSpace()\label{line:projtoss}\\	
			\uIf{$I^{(p)} \subseteq I^{(p+1)}$}   { \label{line:proj1}
				$[H_x^{\mathcal{I}}, h_x^{\mathcal{I}}]$ := PolytopeToMatrix($I^{(p)}$) \label{line:proj2}\\
				STOP successfully.\\
			}
			\uElseIf{$I^{(p+1)}$ is empty}{
				STOP unsuccessfully.\\
			}
			\Else{
				$p:=p+1$\\
			}    
		}
		\KwRet\ $H_x^{\mathcal{I}}, h_x^{\mathcal{I}}$\\
	}
	\caption{Computing the invariant subset $\mathcal{I}$.}
	\label{alg:invariant}
\end{algorithm}

In this algorithm, matrix $H^a_x$ and vector $h^a_x$ represent the adjusted safety region $\mathcal{S}'$ as a polytope (cf. Subsection \ref{sec:safetydefinition}), and matrix $H_x^\mathcal{I}$ and vector $h_x^{\mathcal{I}}$ represent $\mathcal{I}$ as a polytope. Matrix $H_u$ and vector $h_u$ also represent a polytope for operational ranges of control inputs. $A_d^{(m+1)}$ and $B_d^{(m+1)}$ are the matrices to find the state after $m+1$ cycles \ie $x[k+m+1] = A_d^{(m+1)}x[k] + B_d^{(m+1)}u[k]$, where $m = \uptau_r/\uptau_c$. We have $A_d^{(m+1)} = (A_d)^{m+1}$ and $B_d^{(m+1)} = (A_d^{m} + A_d^{m-1}+ ... + I )B_d$.

Intuitively, this algorithm starts from $\mathcal{S}'$ as initial region (line~\ref{line:adjust}). In every iteration of this algorithm, this region is augmented in the extended state-control space $\mathbb{R}^{n+m}$~(line~\ref{line:stateexpansion}). This linear inequality is then projected back into the state space~(line~\ref{line:projtoss}). The outcome of lines~\ref{line:stateexpansion} and~\ref{line:projtoss} is to calculate $\mathcal{I}^{(p+1)}$ which is a subset of $\mathcal{I}^{(p)}$ where a control value in $\mathcal{S}_u$ exists such that, the state in one cycle and $m+1$ cycle after are inside $\mathcal{I}^{(p)}$.

The algorithm proceeds until either $\mathcal{I}^{(p)} \subseteq \mathcal{I}^{(p+1)}$ or $\mathcal{I}^{(p+1)} = \emptyset$. In the former case, procedure successfully ends~(lines~\ref{line:proj1} to~\ref{line:proj2}). 
The latter case indicates that the dynamics of the system does not allow such a region, for the given restart time. There are cases in which the procedure does not terminate in a finite number of steps unless a finite $p_{max}$ is fixed. This may happen if $\mathcal{I}^{(\infty)}$ has an empty interior, but it is not empty~\cite{blanchini2008set}.

If matrix $A_d$ and $B_d$ are controllable, we can use ideas from~\cite{rungger2016computing} to ensure convergence. However, in general, we cannot guarantee that the procedure in Algorithm~\ref{alg:invariant} will converge to a non-empty $\mathcal{I}$. In such cases, one may have to loosen the safety constraints of the system~(\ie $\mathcal{S}$) or may have to switch to a hardware platform with a shorter restart time, to be able to apply this approach.

\subsubsection{Base Controller in Runtime}
\label{sec:bcruntime}
Using invariant set $\mathcal{I}$ computed as given in Subsection~\ref{sec:iregion}, one can compute a control input $u[k]$ at $k$th sampling instance that satisfies the following conditions: 

\begin{align}
\begin{array}{rl}
H_u \cdot u[k]  & \leq h_u, \\
H_x^\mathcal{I}  \cdot x[k+1] & \leq h_x^\mathcal{I},\\ 
H_x^\mathcal{I}  \cdot x[k+m+1] & \leq h_x^\mathcal{I}.\\
\end{array}
\end{align}
By substituting $x[k+1]$ and $x[k+m+1]$, we get the following linear matrix inequalities: 
\begin{equation}
\begin{array}{rl}
H_u u[k] & \leq h_u,\\
H_x^\mathcal{I} B_d u[k] & \leq h_x^\mathcal{I} - H_x^\mathcal{I}  A_d  x[k],\\
H_x^\mathcal{I} B_d^{(m+1)} u[k] & \leq h_x^\mathcal{I} - H_x^\mathcal{I} A_d^{(m+1)}x[k].\\
\end{array}
\label{eq:scruntime}
\end{equation}
At runtime, BC receives the sensors values \ie $x[k]$, and calculates $u[k]$ by solving these inequalities.

\section{Case study and Evaluation}
\label{sec:evaluation}


To demonstrate the practicality of the proposed approach, we implemented a controller for two benchmark systems: (i) inverted pendulum system and (ii) 3-DOF helicopter\cite{3dofhelicopter} and empirically verify fault-tolerance guarantees. We utilize one COTS platform to implement our controller. We inject faults in the control logic, control application, and the operating system to demonstrate that the system remains safe, despite the faults, and recovers. 

\subsection{Experimental Setup}

For the prototype of the proposed design, an i.MX7D application processor is used. This SoC provides two general purpose ARM Cortex-A7 cores capable of running at the maximum frequency of 1 GHz and one real-time ARM Cortex-M4 core that runs at the maximum frequency of 200MHz. The real-time core runs from tightly coupled memory to ensure predictable behavior required for the real-time applications/tasks. The real-time core of the considered platform runs FreeRTOS~\cite{freertos}, an operating system for real-time applications. Because our control tasks have real-time constraints, we implement our controller on the real-time core. Ideally, the general purpose cores would have been completely disabled for the experiments. However, in i.MX7D platform, only Cortex-A7 cores have direct access to the flash memory and, only these two cores can load the binary images of the real-time core from flash into the real-time core's memory after each restart. Hence, instead of permanently disabling those cores, they are only disabled after the software of the real-time core is loaded from flash into the memory. Note that, this mechanism is specific to this particular platform and does not impact the generality of our proposed technique.




The manufacturer's boot procedure of the board is designed to boot the general purpose cores and the real-time core at the same time. It includes extra initialization procedures that are necessary only for running the general purpose core's kernel and mounting its file system. It loads the real-time core code only after those procedures are completed.

To reduce the boot time of the real-time core, we made two modifications to the bootloader~(u-boot) source code which can be found in~\cite{githubrepo}. (i) We included the binary of the real-time core executables~(FreeRTOS, MC, BC, DM, and flushing task) as a static array in the u-boot source code and made it part of the u-boot binary after compilation. (ii) In our modified boot process, at the boot time, the general purpose processor copies u-boot binary~(that includes the FreeRTOS and application binaries) from the SD-card into the RAM. After successful initialization of only the necessary peripherals and configuring the clock by the u-boot procedures, u-boot loads the binaries of the real-time core in its tightly coupled memory and releases it from reset. These modifications reduce the real-time core's boot time from seconds to less than 250ms\footnote{BC task activates one of the GPIO pins immediately after it executes. The restart time is measured externally using the signal on this pin. After multiple experiments, a conservative upper bound was picked for the restart time.}.

\subsection{Example 1: Inverted Pendulum}
\begin{figure}[t]
	\begin{center}
		\includegraphics[width=0.55\textwidth]{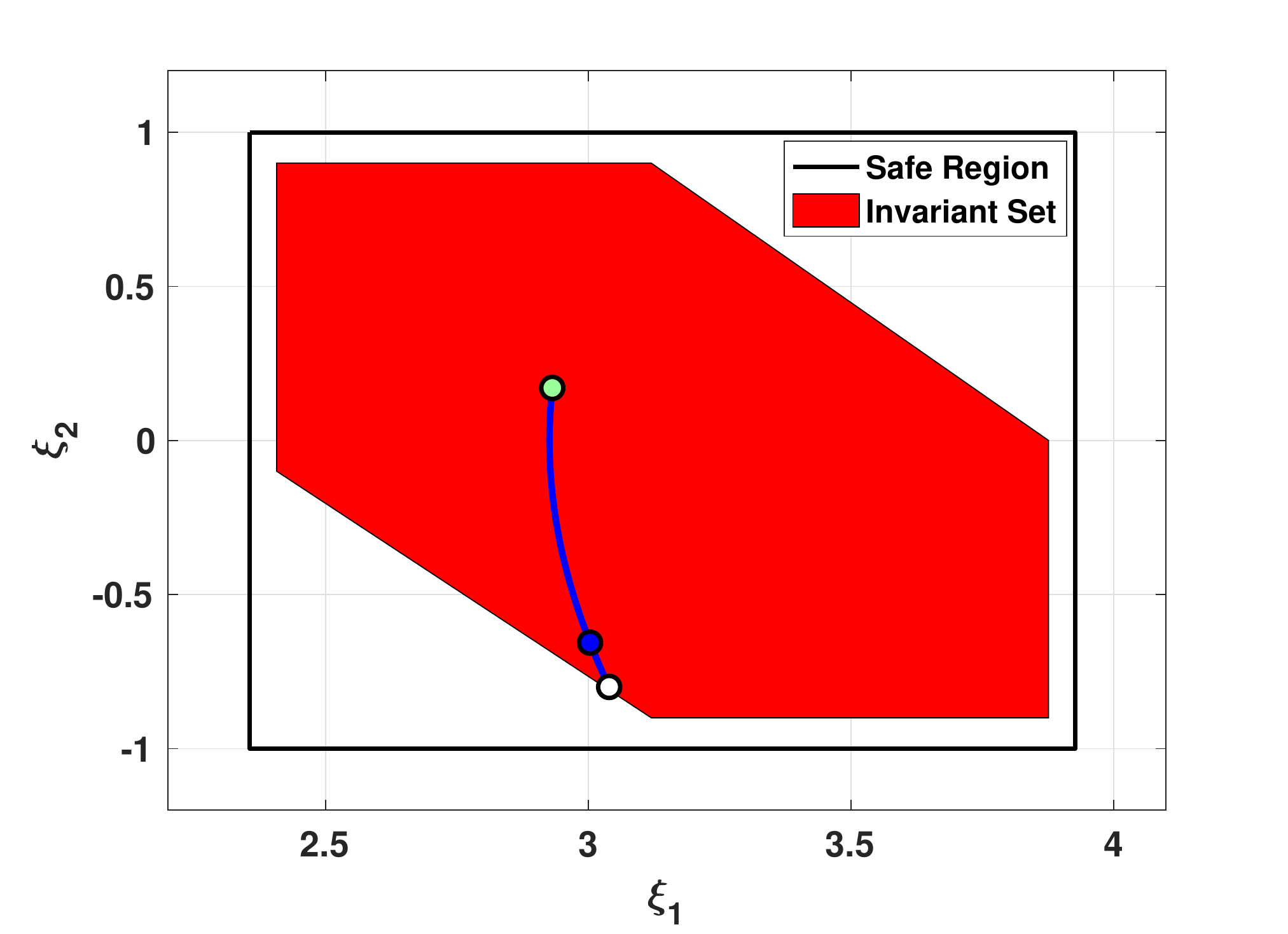}
		\caption{Invariant set obtained using abstraction based approach and a simulated closed-loop trajectory of the system under $u = 3$ which is inside $\mathcal{I}$~(red region) at times $\uptau_c = $~50~ms~(blue mark) and $\uptau_c + \uptau_r = $~300~ms~(green mark). White circle marks the beginning of the trajectory $\xi=[3.04, -0.8]^T$.}
		\label{fig:BC_abstraction}
	\end{center}
\end{figure}
For the first case study, we consider a nonlinear inverted pendulum system \cite{5770194} given by nonlinear differential equations as:
\begin{align}
\dot{\xi}_1(t)&=\xi_2(t)\nonumber\\
\dot{\xi}_2(t)&=-\omega^2\big(\sin(\xi_1(t))+\cos(\xi_1(t))\upsilon(t)\big)-2\gamma \xi_1(t),
\label{pendulum}
\end{align}
with parameters $\omega=1$ and $\gamma=0.0125$. The states $\xi_1$ and $\xi_2$ are the angular position with respect to a downward vertical axis and the angular velocity of the pendulum, respectively. The control input $\upsilon(t)$ is restricted to $[-4, 4]$. We design MC as $\upsilon_{mc}=2(\pi-\xi_1-\xi_2)$ to stabilize the pendulum at upright position that is $\xi=[\pi,0]^T$ which runs with frequency of 20Hz (i.e. $\uptau_c=$50ms) on real-time core of i.MX7D. To ensure safety of the system (i.e. to avoid pendulum to fall down), we consider safety region for the states given by a polytope parameterized by\\
\begin{align*}
H_x=\begin{bmatrix}
-1&0\\
1& 0\\
0&-1\\
0& 1
\end{bmatrix}
\text{ and }
h_x=\begin{bmatrix}
-0.75\pi\\ 1.25\pi\\ 1\\ 1
\end{bmatrix}.
\end{align*}
To ensure fault-tolerance and safety during restart, we designed BC using abstraction-based approach as discussed in Section \ref{BC}. To synthesize BC, we first constructed a discrete abstraction of the pendulum system in (\ref{pendulum}) using quantization parameter $\eta=[0.05,0.1]^T$, sampling time $\uptau_c=0.050$, and restart time $\uptau_r=0.250$.
Further, we synthesize a safety controller using maximal fixed point computation algorithm. For the controller synthesis, we used toolbox SCOTS \cite{rungger2016scots} with some modifications to adapt the construction of abstraction given in Subsection \ref{Disc_abs} . The invariant states computed using the proposed approach is shown in Figure \ref{fig:BC_abstraction}. To verify the efficacy of the designed controller, we implemented it on our experimental setup (i.MX7D) and tested in the closed-loop with inverted pendulum dynamics simulated in the computer under various test scenarios discussed in Subsection \ref{faults}.   

\subsection{Example 2: 3-DOF Helicopter}
3-DOF helicopter~(displayed in Figure~\ref{fig:3dof}) is a simplified helicopter model, ideally suited to test intermediate to advanced control concepts and theories relevant to real-world applications of flight dynamics and control in the tandem rotor helicopters, or any device with similar dynamics \cite{3dofhelicopter}. It is equipped with two motors that can generate force in the upward and downward direction, according to the given actuation voltage. 
It also has three sensors to measure elevation, pitch, and travel angle as shown in Figure~\ref{fig:3dof}.
We use the linear model of this system obtained from the manufacturer manual \cite{3dofhelicopter}. The BC is designed as discussed in Subsection \ref{BC_linear}. 
\begin{figure}[ht]
	\begin{center}
		\includegraphics[width=0.35\textwidth]{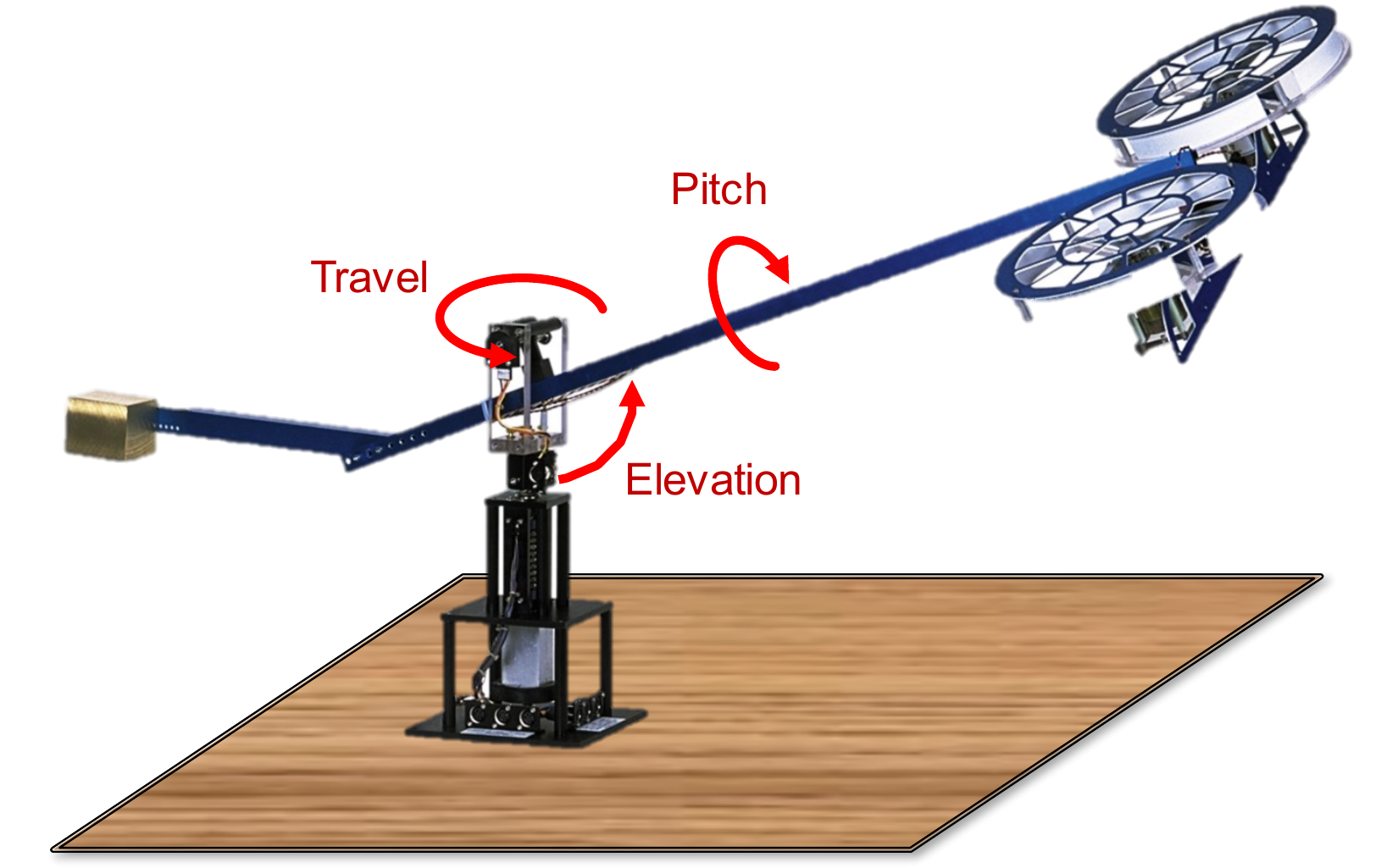}
		\caption{3 Degree of freedom~(3-DOF) helicopter.}
		\label{fig:3dof}
	\end{center}
\end{figure}

For 3-DOF helicopter, the safety region is defined in such a way that the helicopter fans do not hit the surface underneath, as shown in Figure~\ref{fig:3dof}, while respecting the maximum angular velocities. The six dimensional state vector is given by $x=[\epsilon,\rho,\lambda,\dot{\epsilon},\dot{\rho},\dot{\lambda}]^T$,  where variables $\epsilon$, $\rho$, and $\lambda$ are the elevation, pitch, and travel angles, respectively, $\dot{\epsilon}$, $\dot{\rho}$,and $\dot{\lambda}$ are the corresponding angular velocities. The $u=[v_l,v_r]^T$ represents input vector, where $v_l$ and $v_r$ are the voltages applied to right and left motors. The safe region for the state and input spaces are represented using polytopes as discussed in Subsection \ref{sec:safetydefinition} and parametrized with
\begin{align*}
H_x=\begin{bmatrix}
-1 & -0.33 & 0 & 0 & 0 & 0 \\
-1 & 0.33 & 0 & 0 & 0 & 0 \\
0 & 0 & 0 & 1 & 0 & 0 \\
0 & 0 & 0 & -1 & 0 & 0 \\
0 & 0 & 0 & 0 & 1 & 0 \\
0 & 0 & 0 & 0 & -1 & 0 
\end{bmatrix},
h_x=\begin{bmatrix}
0.3\\
0.3\\
0.4\\
0.4\\
1.5\\
1.5
\end{bmatrix},
H_u=\begin{bmatrix}
-1  & 0 \\
-1  & 0 \\
0 & 1  \\
0 & -1  
\end{bmatrix}, \text{and }
h_u=\begin{bmatrix}
1.1 \\
1.1 \\
1.1 \\
1.1 
\end{bmatrix}.
\end{align*}




FreeRTOS on the Cortex-M4 core restarts in 250 ms (upper bound). By using algorithm in \cite[Section 5.1]{7945009}, we computed readjusted safety constraint parameters as $ h^a_x = [0.1418$, $0.1418$, $0.2828$, $0.2828$, $0.0825$ ,$0.0825 ]^T$ and $H^a_x=H_x$.
Using this readjusted safety constraints the invariant region and BC are constructed using the Algorithms described in Subsections~\ref{sec:iregion} and~\ref{sec:bcruntime}. Algorithm~\ref{alg:invariant} computed a region $\mathcal{I}$ confined with 106 inequalities after 14 iterations. The offline computation took four hours on Mac Book Pro with 2.5 GHz Intel Core i7 and 16 GB of memory. Finally, the BC is derived by solving linear inequalities in (\ref{eq:scruntime}).
\subsubsection{Hardware Interface}
The control tasks on the real-time core of i.MX7D run with a frequency of 20~Hz~($\uptau_c =$ 50~ms). Our controller interfaces with the 3DOF helicopter through a PCIe-based \textit{Q8 High-Performance H.I.L. Control and data acquisition unit}~\cite{q8daq} and an intermediate Linux-based PC. The PC communicates with the i.MX7D through the serial port. At the end of every control cycle, a flushing task on the real-time core communicates with the PC to receive the sensor readings~(elevation, pitch, and travel angles) and send the motors' voltages. It also updates the hardware WD of the platform after sending the motor voltages. The PC uses a custom driver written for Linux to send the voltages to the 3DOF helicopter motors and reads the sensor values. 
\subsubsection{Testing the Base Controller}
To verify that the constructed base controller has the desired properties, we simulated the system with this controller from all vertices of region $\mathcal{I}$ as starting points and observed that the system's state at $\uptau_c$ and $\uptau_c + \uptau_r$ time units after actuation was inside $\mathcal{I}$. Figure~\ref{fig:bcsimulations} outlines one extreme example. The trajectory starts at 
$\epsilon = -0.1410$,
$\rho = 0$,
$\dot{\epsilon} = -0.0281$ and
$\dot{\rho} =  0.0513$~($\lambda$ and $\dot{\lambda}$ do not impact safety). The control command in this trajectory is $v_r = 0.6863$ and $v_l = 0.7709$. As shown in Figure~\ref{fig:bcsimulations}, the trajectory remains inside the safety region.
\begin{figure}[pht]
	\begin{center}
		\centering 
		\subfigure[Projecttion to $\epsilon$ and $\rho$]
		{\label{fig:11}\includegraphics[width=0.24\textwidth, height=3cm]{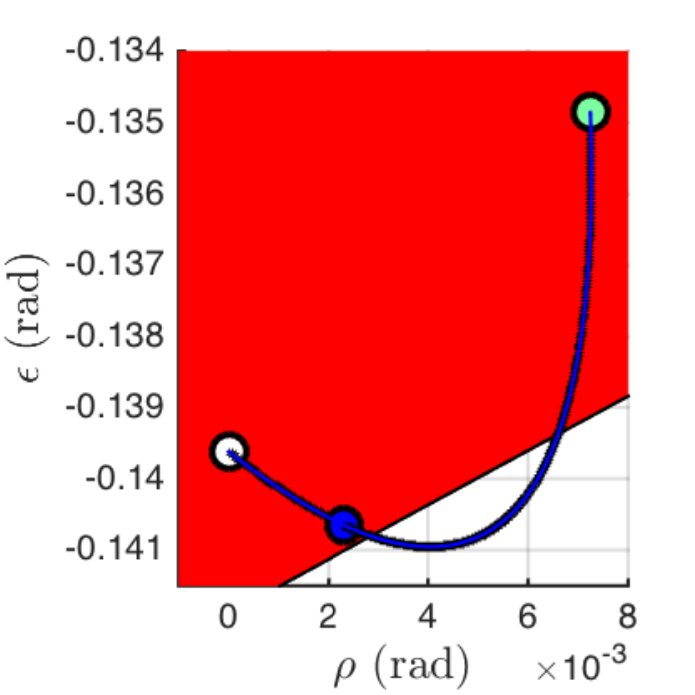}}
		\subfigure[Projecttion to $\epsilon$ and $\dot{\rho}$]
		{\label{fig:O12}\includegraphics[width=0.25\textwidth, height=3cm]{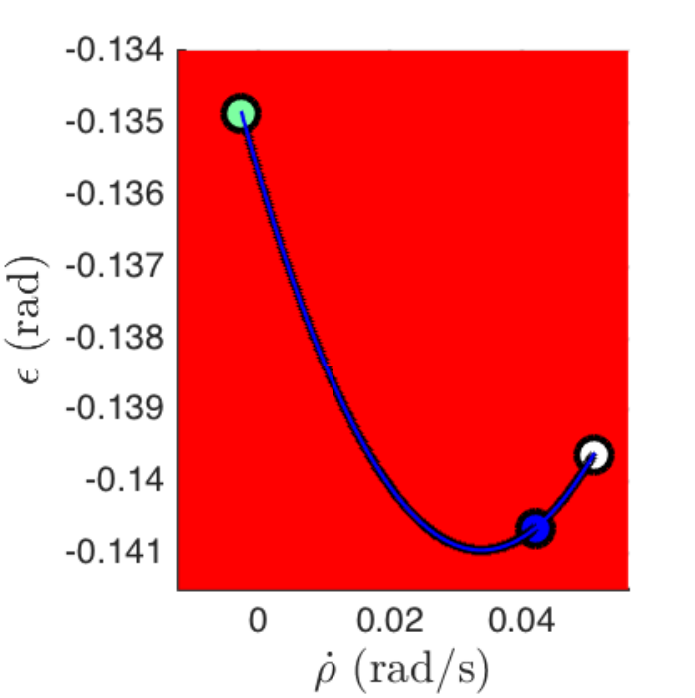}}
		\subfigure[Projecttion to $\dot{\epsilon}$ and $\rho$]
		{\label{fig:21}\includegraphics[width=0.22\textwidth, height=3cm]{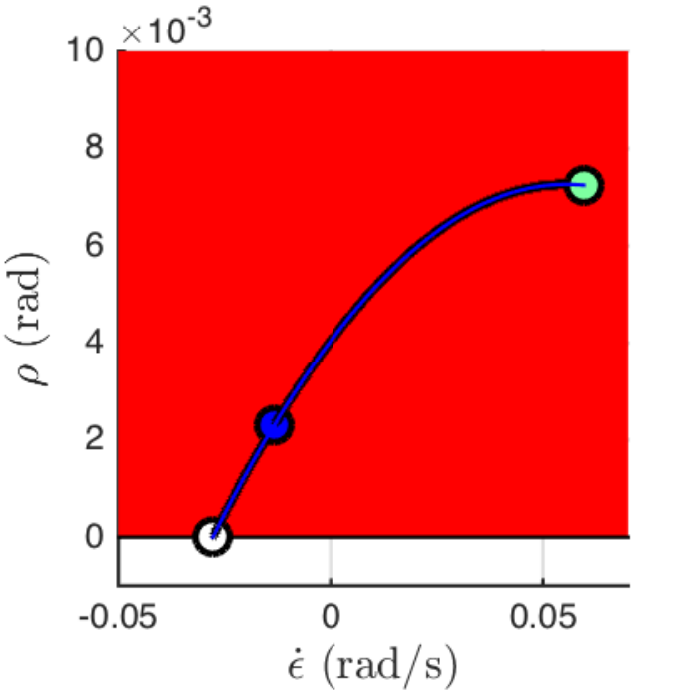}}
		\subfigure[Projecttion to $\rho$ and $\dot{\rho}$]
		{\label{fig:22}\includegraphics[width=0.22\textwidth, height=3cm]{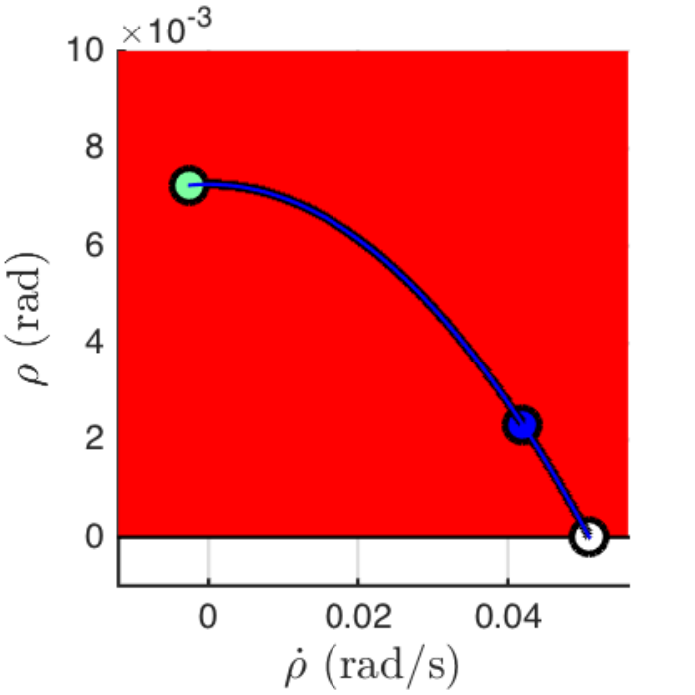}}
		\caption{Simulated trajectory of the system under $v_r = 0.6863$ and $v_l = 0.7709$ is inside $\mathcal{I}$~(red region) at times $\uptau_c = $~50~ms~(blue mark) and $\uptau_c + \uptau_r = $~300~ms~(green mark). White circles mark the beginning of the trajectory. The trajectory is projected into the four planes for clarity.}
		\label{fig:bcsimulations}
	\end{center}
\end{figure}
Further, we implemented the obtained BC on the i.MX7D platform to validate our design approach under different fault scenarios given in the next section.

\subsection{Fault Injection}
\label{faults}
In Table~\ref{table:tests}, a list of faults that were tested on the implementations are provided. We also compare them with Application-Level Simplex and System-Level Simplex. For the application-level faults, we verified that the mission controller was able to actuate the system as long as it did not jeopardize the safety and when the system states approached the states where the safety conditions violated, BC took over and ensure safety. For the system-level faults, we observed that the WD restarted the system and after restart, the system continued its operation. 
\begin{table*}[]
	\centering
	\small 
	\scalebox{0.8}{
		\begin{tabular}{|c|c|c|c|c|c|}
			\hline
			\cellcolor[HTML]{C0C0C0}{\color[HTML]{333333} }                                        & \cellcolor[HTML]{C0C0C0}{\color[HTML]{333333} }                                                                                    & \multicolumn{3}{c|}{\cellcolor[HTML]{C0C0C0}{\color[HTML]{333333} Safety}}                                                                                                                                                                                                                      & \cellcolor[HTML]{C0C0C0}                                     \\ \cline{3-5}
			\multirow{-2}{*}{\cellcolor[HTML]{C0C0C0}{\color[HTML]{333333} \textbf{Failure Type}}} & \multirow{-2}{*}{\cellcolor[HTML]{C0C0C0}{\color[HTML]{333333} \textbf{\begin{tabular}[c]{@{}c@{}}Fault\\ Category\end{tabular}}}} & \cellcolor[HTML]{9B9B9B}{\color[HTML]{333333} \begin{tabular}[c]{@{}c@{}}Application-Level\\ Simplex\\(Single HW Board/SoC)\end{tabular}} & \cellcolor[HTML]{9B9B9B}{\color[HTML]{333333} \begin{tabular}[c]{@{}c@{}}System-Level\\ Simplex\\ (Additional HW/SoC)\end{tabular}} & \cellcolor[HTML]{9B9B9B}{\color[HTML]{333333} \begin{tabular}[c]{@{}c@{}}Our Approach\\ (Single HW Board/SoC)\end{tabular}} & \multirow{-2}{*}{\cellcolor[HTML]{C0C0C0}\textbf{Restarted}} \\ \hline
			\cellcolor[HTML]{EFEFEF}No Output                                                      & App.                                                                                                                        & \cmark                                                                                                                 & \cmark                                                                                                           & \cmark                                                          & No                                                           \\ \hline
			\cellcolor[HTML]{EFEFEF}Maximum Voltage                                     & App.                                                                                                                        & \cmark                                                                                                                 & \cmark                                                                                                           & \cmark                                                          & No                                                           \\ \hline
			\cellcolor[HTML]{EFEFEF}Time Degraded Control                                          & App.                                                                                                                        & \cmark                                                                                                                 & \cmark                                                                                                           & \cmark                                                          & No                                                           \\ \hline
			\cellcolor[HTML]{EFEFEF}Timing Fault - CPU                                                 & RTOS/App.                                                                                                                     & \xmark                                                                                                                & \cmark                                                                                                           & \cmark                                                          & Yes                                                          \\ \hline
			\cellcolor[HTML]{EFEFEF}Timing Fault - Resource                                                  & RTOS/App.                                                                                                                     & \xmark                                                                                                                & \cmark                                                                                                           & \cmark                                                          & Yes                                                          \\ \hline
			\cellcolor[HTML]{EFEFEF}FreeRTOS Freeze                                                      & RTOS                                                                                                                                 & \xmark                                                                                                                & \cmark                                                                                                           & \cmark                                                          & Yes                                                          \\ \hline
			\cellcolor[HTML]{EFEFEF}Computer Reboot                                                & RTOS                                                                                                                                 & \xmark                                                                                                                & \cmark                                                                                                           & \cmark                                                          & Yes                                                          \\ \hline
	\end{tabular}}
	\caption{ Our approach tolerates system-level faults using only one hardware unit. Whereas, System-Level Simplex~\cite{bak2009system} needs an extra board/SoC to tolerate these faults.}
	\label{table:tests}
\end{table*}


Some of these faults are elaborated in the rest of this section. 


\subsubsection{Maximum Control Input in Wrong Way}
The system should not leave safe region even if the MC outputs a control input that normally would result in a crash. We consider an extreme case of this scenario where the MC generates a control input that forces system towards the unsafe region. The unsafe MC commands were detected by DM~(they did not satisfy the system safety conditions), and the control was switched to the BC until the system was in the safety region and then control was handed back to MC.




\subsubsection{Timing Faults (CPU and Resource)}
The proposed solution also protects the system from timing faults. A faulty task may behave differently in runtime from its expected/reported behavior. For instance, it may lock a particular resource used by other critical tasks for more than the intended duration. Or, it may run for more time than its reported worst-case execution time~(WCET) which was used for the schedulability test of the system. Timing faults may also originate from RTOS or driver misbehaviors. If the fault delays/stops the execution of the DM or BC, WD will trigger a system-wide restart. This recovers the system from the fault and keeps the physical system safe. We perform two experiments to test the fault-tolerance against timing faults.

In the first experiment, we run an additional task on the system that uses the serial port in parallel to the flushing task to communicate with the PC. We inject a fault into this task so that in random execution cycles, it holds the lock on the serial port for more than its intended period. This prevents the flushing task from updating the actuator~(which needs the serial port) before the end of the control cycle. As a result, WD expires and restarts the system. We verified that the system recovers from the fault and remains safe during the restart.


In the second test, we introduce a task that runs at the same priority as the BC and DM. We inject a fault into the task such that in some cycles, its execution time exceeds its reported WCET. FreeRTOS runs the tasks with equal priority using round-robin scheduling with a context switch at every 1ms. Therefore, the faulty task delays the response time of the DM and BC. If the interference is too long, the output of BC may not be ready by the time the flushing task needs to update the actuators. When this happens, WD restarts the system. 

\section{Discussion}
\label{sec:limitations}


\textit{Software Faults}: The proposed approach does not handle software faults that modify the program logic or output of the BC and the DM at the execution time. Utilizing frameworks such as ARM TrustZone~\cite{trustZone} and limiting the access to these critical components can mitigate this issue.

\textit{Restart Time}: As the restart time of the platform increases, the domain of the BC shrinks. Therefore, the proposed solution in its current form, even though useful for many platforms, may not suit some platforms with a longer restart time.

We are actively working on an alternative multi-stage booting solution for multicore platforms to mitigate this problem. Our main idea is to boot one core with the bare minimum requirements to execute the BC in the shortest possible time. The BC can keep the system safe, while the real-time or general purpose OS boots on the other cores. Once the boot process is complete, the control switches to the controllers running on the OS. As a future extension, we are working on implementing this solution on i.MX7D platform.  We first boot the real-time core with a FreeRTOS and run the BC on top of it and then boot an embedded Linux on the general purpose core.

\section{Conclusion}
\label{sec:conclusion}
Restarting is considered as a reliable way to recover the traditional computing system from software faults. However, restarting safety-critical CPS is challenging. In this work, we propose a novel approach that guarantees \textit{safety} and \textit{liveness} of nonlinear physical systems in the presence of application and system-level software faults utilizing only one COTS processor based on complete system-level restarts.







\bibliographystyle{alpha}
\bibliography{sample-bibliography}

\end{document}